\documentclass{article}
\usepackage[a4paper,margin=1in]{geometry}

\usepackage[T1]{fontenc}

\usepackage[usenames,dvipsnames,svgnames,x11names]{xcolor}

\newcommand{\delc}[1]{}

\usepackage[T1]{fontenc} 
\usepackage[normalem]{ulem} 
\usepackage{dsfont}

\usepackage{balance}

\usepackage[nocompress]{cite}

\AtBeginDocument{%
  \providecommand\BibTeX{{%
    \normalfont B\kern-0.5em{\scshape i\kern-0.25em b}\kern-0.8em\TeX}}}

\usepackage{listings}

\usepackage{moreverb}

\usepackage[nounderscore]{syntax}

\usepackage{amssymb}
\usepackage{mathtools}
\usepackage{amsfonts}

\usepackage{bm}
\usepackage{amsthm}

\usepackage{subcaption}
\usepackage{pifont}

\usepackage{lipsum}
\usepackage{algorithmicx}

\usepackage{multirow}
\usepackage{rotating}
\usepackage{booktabs}
\usepackage{colortbl}
\usepackage{tablefootnote}

\usepackage{makecell}

\usepackage{array}
\newcolumntype{L}[1]{>{\raggedright\let\newline\\\arraybackslash\hspace{0pt}}m{#1}}
\newcolumntype{C}[1]{>{\centering\let\newline\\\arraybackslash\hspace{0pt}}m{#1}}
\newcolumntype{R}[1]{>{\raggedleft\let\newline\\\arraybackslash\hspace{0pt}}m{#1}}

\usepackage[colorlinks,bookmarksopen,bookmarksnumbered,citecolor=red,urlcolor=blue]{hyperref}

\usepackage{xspace}
\let\labelindent\relax
\usepackage{enumitem}
\usepackage{blindtext}

\newtheorem{theorem}{Theorem}

\usepackage{url}

\usepackage[linesnumbered,ruled,vlined]{algorithm2e}
\SetKwInput{KwData}{Input}
\SetKwInput{KwResult}{Output}

\begin{document}

\title{
\LARGE \bf
Adaptive and Cost-Efficient Joint Scheduling of UAV Routes and Analytics with Transit-Borne Fog}

\author{Suman Raj, Arindam Khanda, Gagana M D, Yogesh Simmhan, and Sajal K. Das
\thanks{Raj is with the Department of Computer Science, University of Chicago, USA (sumanraj@uchicago.edu). Khanda is with the Department of Information Technology, Kennesaw State University, USA (akhanda@kennesaw.edu). Das is with the Department of Computer Science, Missouri University of Science \& Technology, USA (sdas@mst.edu). Simmhan is with the Department of Computational and Data Sciences, Indian Institute of Science (IISc), Bangalore, India (simmhan@iisc.ac.in). Raj and Gagana (gagana.md.work@gmail.com) were with IISc at the time of writing this paper. } 
} 

\date{}

\maketitle
\thispagestyle{plain}
\pagestyle{plain}

\begin{abstract} 
Unmanned Aerial Vehicles (UAVs) performing deadline-bound analytics over large rural areas cannot reliably offload workloads to sparse cellular base stations. We propose an approach that uses scheduled public buses as \textit{mobile fogs}: a UAV hands off data to a bus during a halt, and the bus carries it until its route enters cellular coverage. Because a public bus follows a fixed route and timetable, a handover depends on \textit{where} and \textit{when} the bus next reaches a cellular zone, rather than how near the stop is. 
We formulate a Mission Scheduling Problem over this model, co-scheduling UAV routes with the placement of each analytics task on the UAV edge, a stationary fog, or a bus, under deadline, energy, and cost constraints.
Our \textit{Divide and Assign} (DA) heuristic selects the cheapest halt that still meets a task's deadline.
Across 36 workload configurations in a rural region, derived from real cellular and transit data, DA achieves up to 20\% higher utility than the strongest heuristic and up to 41\% higher utility than the strongest adapted-prior scheduler, while incurring the lowest aggregate cost.
The transit tier handles up to 71\% of drop-offs, raises task completion to 100\%, and reduces recharging cycles by up to 31\%. Finally, in the presence of traffic variability, the adaptive variant recovers 93\% of the utility the delays cost and maintains completion rates above 97\%.

\end{abstract}

\setlength{\belowdisplayskip}{0pt} \setlength{\belowdisplayshortskip}{0pt}
\setlength{\abovedisplayskip}{0pt} \setlength{\abovedisplayshortskip}{0pt}

\section{Introduction}

Unmanned Aerial Vehicles (UAVs), or drones, are increasingly used for spatio-temporal sensing over large rural areas, such as crop monitoring, utility surveys and construction inspection~\cite{betti2024drone}. Much of this sensing is only useful if it is acted upon promptly: a sparking distribution line or a pest outbreak should raise an alert during the flight rather than after the fleet lands. Edge accelerators are small enough to fly on board. For instance, an NVIDIA Jetson Orin Nano offers $1{,}024$ Ampere CUDA cores and $8$\,GB of memory within $15$\,W and weighs under $200$ grams. 
However, on-board analysis competes with other flight operations: a quadcopter's motors draw $\approx750$\,W against a $20$--$45$~min endurance, so every second of inference is a second of airtime and a joule of battery. UAVs cannot analyze everything they sense, and the surplus must be \textit{offloaded}~\cite{mendula2024furcifer}.

\vspace{2pt}\noindent\textbf{The rural offloading gap.}
Dense urban cellular infrastructure often permits direct offloading to nearby edge or cloud resources, whereas rural missions may operate across large areas with sparse coverage.
Compute co-located with $4$G/$5$G base stations, which we call \textit{stationary fogs}, is sufficiently dense in major urban regions, as shown in Fig.~\ref{fig:motivation}. Rural airspace is the opposite regime. This leaves the agricultural interior, precisely where the sensing happens, thinly served. Flying to the nearest tower or deferring everything to the base station wastes the very flight time and energy that the deadline is competing for.
\textit{The instinct is to deploy more infrastructure.} We take the pervasive-computing route instead and leverage what already passes through the landscape on a published schedule: the rural bus. 
Public bus services in developing and developed regions already use onboard connectivity and, in some deployments, onboard compute~\cite{MyCitiBus2025,BMTC_Official} allowing suitably equipped buses to act as mobile fogs without deploying additional field infrastructure.

\begin{figure}[t]
    \centering
    \includegraphics[width=0.4\linewidth]{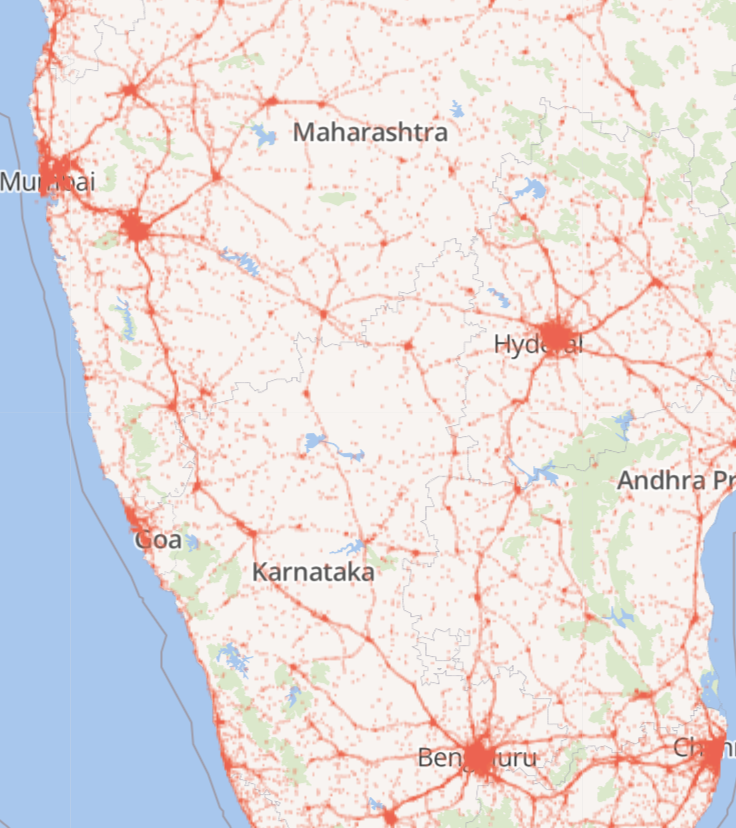}
    \caption{Cellular coverage across peninsular India (source: OpenCelliD~\cite{OpenCelliD}). Coverage clusters in cities and strings out along highways, leaving the agricultural interior thinly served; those same roads carry the bus routes we enlist as mobile fogs.}
    \label{fig:motivation}
\end{figure}

This goes beyond a fleet optimization problem. The provider owns the UAVs but not the buses; it cannot move a stop, hold a departure, or request capacity. It can only read a timetable and schedule the drone to be in the right place at the right moment, so the physical and temporal structure of an unowned public service becomes a hard input to the computation. Fig.~\ref{fig:overview} sketches the resulting mission.
A UAV leaves its depot, Observes (O) a field at rural Location~1, then flies to a bus stop and offloads the data (D1) over WiFi during the bus's halt. The bus Analyses (A) it en route and, on entering cellular range, relays the result via $s_2$ (a cell tower) to the cloud. Meanwhile, the UAV moves on to pick up (P) sensor data (D2) from a field sensor at Location~2, analyses it onboard while transiting to Location~3 for a further pickup (D3), and delivers the result of D2 and the raw D3 to the stationary fog $s_1$. The bus is thus both a compute tier and a \textit{data mule}, a transient fabric bridging the UAV edge to the stationary fogs and the cloud.

\vspace{2pt}\noindent\textbf{Why this is hard.}
Tasks are deadline-bound, e.g. detecting a cut cable during a survey, and where task can run depends on the model it needs: YOLOv8-nano fits on the UAV, a vision transformer does not. Resources price differently, an RTX 3090 fog costing about US\$$48\times10^{-6}$/sec amortised over three years against US\$$6.6\times10^{-6}$/sec for an Orin Nano, and cellular offload adds data charges that a WiFi handover to a bus avoids. Energy binds throughout, since a $20$--$45$~min endurance must cover flying, hovering, inference, and any diversion to recharge.
The fourth is specific to this setting and drives the model. The bus schedule is an input, not a decision variable, so a UAV that reaches a stop after the bus has pulled away has spent the flight for nothing. Moreover, the payload rides an itinerary the scheduler does not control, and surfaces only when that route next enters cellular range. 
Therefore, the value of a halt depends on both the UAV's \textit{handover time} and the bus's subsequent \textit{cellular-reentry time}, rather than on UAV-to-halt distance alone.

\begin{figure}[t]
    \centering
\includegraphics[width=0.65\linewidth]{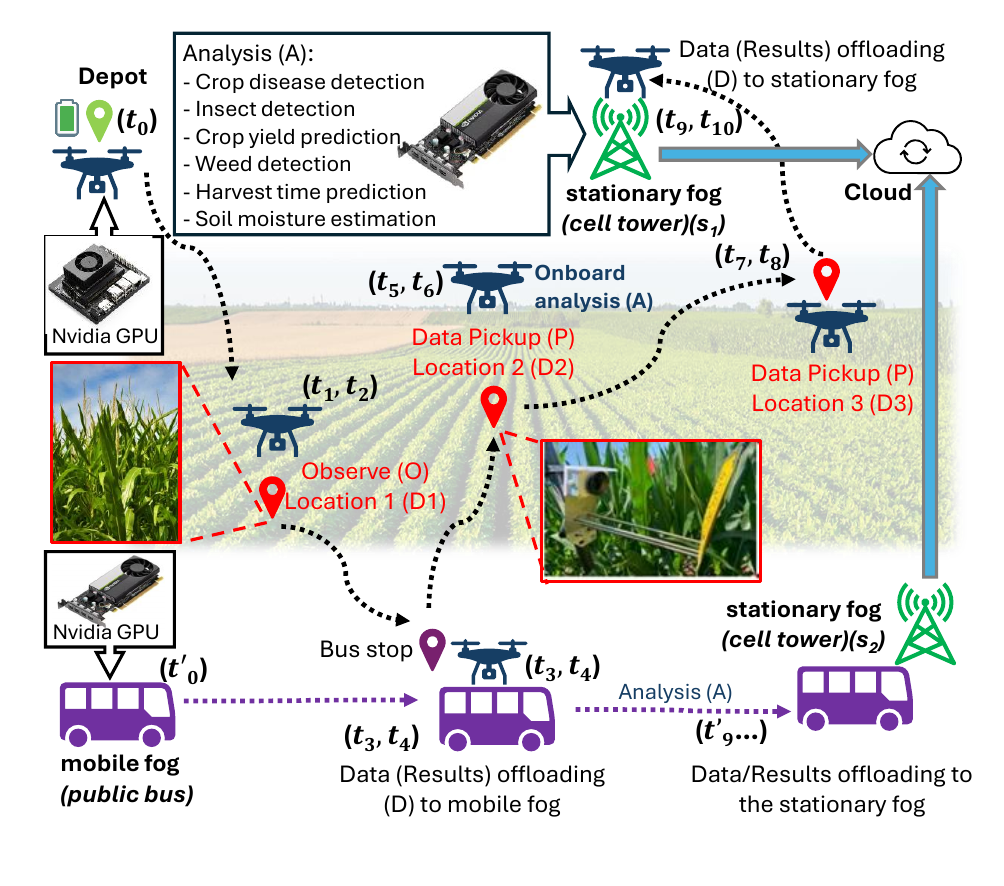}
    \caption{A rural mission across compute tiers: onboard UAV edge, stationary fog at cell towers, and mobile fog on scheduled buses that relay on entering coverage.}
    
    \label{fig:overview}
\end{figure}

\vspace{2pt}\noindent\textbf{Contributions.}
This paper studies mission planning for a UAV fleet operator offering \textit{drone analytics as a service} in remote regions. 
A user-requested \textit{task} may involve a drone capturing video while hovering at a waypoint \textit{(observing)}, collecting data from ground sensors \textit{(pickup)}, running \textit{analytics} on the acquired data, and \textit{offloading} raw or processed data onward. 
Tasks accrue benefit only if delivered within their deadline, and consume monetary cost as they run. This frames our problem: \textit{how do we co-schedule drone routes among waypoints, and place each analytics task on the UAV edge or offload it to a stationary or mobile fog, so that the fleet maximises accrued utility within each drone's energy budget and each task's latency constraint?} We make the following specific contributions.

\begin{enumerate}[leftmargin=*]
    \item We characterize the \textit{system}, \textit{application}, and \textit{latency, energy and cost} models, including a \textit{halt model} in which each halt is valued by its drop-off time $D(c)$ 
    derived from published bus timetables and base stations locations (\S~\ref{sec:models-and-assumptions}).
    \item We formally define the \textit{Mission Scheduling Problem (MSP)} over this model, co-scheduling routes and analytics for a fleet with heterogeneous edge and stationary/mobile fog compute, and prove it \textit{NP-hard} (\S~\ref{sec:problem-formulation}).
    \item We develop
    a scalable and adaptive heuristic, \textit{Divide and Assign (DA)}, which routes each drone over spatio-temporal task clusters and drops each sub-trip at the cheapest destination, that still meets the deadline (\S~\ref{sec:heuristics}).
    \item 
    We \textit{evaluate} DA against  EDF~\cite{stankovic1998deadline}, NWF~\cite{perez2013solving}, and adaptations of three recent schedulers, \textit{PBTO}~\cite{park2024public}, \textit{EOFO}~\cite{10146342} and \textit{TGTD}~\cite{li2026latency}, using a common utility, cost, deadline, and energy accounting model. The region is built end-to-end from open data, with stationary fogs from OpenCelliD, halts and road geometry from a public transit feed, and real DNN models profiled on hardware NVIDIA GPUs (\S~\ref{sec:evaluation}).
\end{enumerate}

\vspace{-0.1in}
\section{Related Work}\label{sec:related-work}

\subsection{Mobility and Task Offloading}
A substantial body of research has focused on offloading computations from edge devices to static fog or cloud~\cite{pournaropoulos2023supporting,sharma2023deep,10.1145/3529706.3529708}. In the context of mobile edges, MoDEMS~\cite{9796680} addresses the need to offload data from mobile edge to different servers for long-term user mobility and formulates cost minimization as an ILP. In~\cite{tang2020deep}, a Deep Reinforcement Learning (DRL) is adapted for offloading tasks from mobile edge to servers.
Some~\cite{10171496,raj2025adaptive} offload DNN tasks from edge to cloud while adapting to network variability due to mobility. 
Neither consider offloading to a mobile fog. 
A movable unmanned ground vehicle (UGV)-mounted mobile edge computing (MEC) server-assisted data acquisition and processing system for UAVs is proposed in~\cite{tang2025uav}. However, this model does not optimize UAV routes or account for proper task deadlines.
SLIM+~\cite{huang2026slim+} jointly optimizes the placement of ground vehicles and UAV fleet sizing to minimize deployment cost. 
Likewise, a heterogeneous graph RL policy is used in~\cite{wu2026ugv} to allocate UAV tasks with dispatchable UGV assistance. 
In contrast, public transit buses follow fixed routes and timetables, which makes the choice of \textit{which halt} to hand off at the central decision in our paper.

\vspace{-0.05in}
\subsection{Public Transit as a UAV Resource}

Recently, a growing body of work treats public buses as the ground resource for UAVs, motivated by their predictable mobility and public-service operation. The authors in~\cite{park2024public} proposed a public bus-assisted UAV task offloading scheme 
over Seoul bus data and maps. This work demonstrates the potential of public buses as mobile fog resources; however, it focuses on dense urban deployments and energy-delay optimization, whereas we address rural settings.
A second line of work exploits transit for \textit{energy} rather than computation. In~\cite{zhu2025effective}, the UAVs land on buses to recharge, and co-solve UAV trajectory planning alongside time-slot assignment of UAVs to points of interest under energy and coverage constraints, to extend system lifetime. Earlier work similarly rides buses to extend UAV range for on-demand monitoring~\cite{gao2022leveraging}. These works establish that transit networks are a viable substrate for UAV operations and, like us, plan routes against a fixed timetable. However, the bus is a \textit{carrier and charger}, not a compute- or data-carrying node, and they optimize coverage and lifetime rather than deadline-bound analytics utility. Complementarily, feasibility studies on bus-mounted edge servers~\cite{li2025busfeasible} and simulation support for UAV-integrated vehicular fog~\cite{wei2026airfogsim} corroborate the practicality of the transit-borne fog substrate that we assume.

\vspace{-0.05in}
\subsection{UAV Related Co-optimizations}
Co-scheduling of task analytics and UAV route planning has also been considered. The authors in~\cite{hao2024joint} jointly optimized UAV trajectories and task offloading in multi-UAV MEC using DRL, while several heuristics are developed in~\cite{10146342} for deadline-aware UAV routing and onboard task scheduling.
However, none of these consider task offloading to a fog. A three-layer computing architecture is proposed in~\cite{sun2024joint} for post-disaster rescue, integrating MEC and vehicle fog computing.
Unlike our approach proposed in this paper, they do not consider task deadlines, and their fog vehicles remain in the UAVs' vicinity without a fixed schedule.
Two recent works co-optimize a UAV waypoint visiting order with its offloading decisions in inspection settings that structurally resemble ours. \cite{guo2026joint} jointly plans trajectories and offloading over sensor clusters, minimizing a weighted latency-energy sum, while~\cite{li2026latency} co-optimizes visiting order, offloading, and resource allocation under deadline and energy budgets. However, they offload to a high-altitude platform cloud in~\cite{li2026latency}, making the offload target a capacity rather than a spatio-temporal decision like ours. 
In~\cite{gao2026cooperative}, UAVs with crowdsourced taxis are coordinated for instant delivery; however, unlike the taxis, which allow a \textit{detour}, our public bus does not allow any. 

In contrast, existing approaches either offload UAV tasks to buses without jointly planning the UAV mission or co-schedule routes against stationary or dispatchable compute tiers; we instead jointly select UAV routes, analytics placement, and handover halts using published bus timetables.

\vspace{-0.05in}
\section{Use Case and Representative Workflows}\label{sec:workflows}

\begin{figure}
    \centering
    \includegraphics[width=0.7\columnwidth]{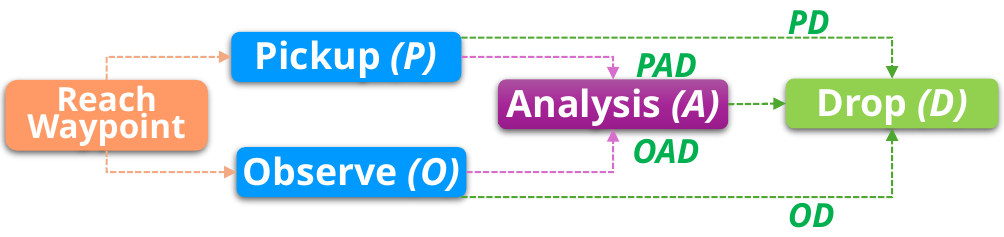}
    \caption{Workflow for different Task Types.}
    \label{fig:state-diagram}
\end{figure}

Satellite imagery and manual field visits sample farmland infrequently and at a high labor cost, whereas a UAV fleet with onboard cameras and DNN accelerators can estimate soil moisture or identify pest diseases far more often~\cite{petchiammal2023paddy,mavridou2019machine}. The value of such an estimate decays with delay, and at very different rates: a yield survey tolerates overnight batch processing, while precision irrigation tracks moisture that shifts within tens of minutes and corn rootworm spreads fast enough that a day's delay is measured in lost crop~\cite{gupta2024analyzing}.

The UAV fleet operator receives \textit{tasks} from users, each belonging to one of four \textit{types} (Fig.~\ref{fig:state-diagram}): \textit{Pickup Drop (PD)}, \textit{Observe Drop (OD)}, \textit{Pickup Analyse Drop (PAD)} and \textit{Observe Analyse Drop (OAD)}. A \textit{pickup} flies the drone to a waypoint to \textit{collect data} from a ground sensor, whereas an \textit{observe} has it \textit{hover} there for a specified duration to \textit{capture} sensor data. The user may additionally request \textit{analysis} over that data before it is \textit{dropped off} to the cloud through a fog. Where that analysis runs, and where the data is dropped off, yields the four orchestration scenarios of Fig.~\ref{fig:workflow-pd-od}. \textit{W1} offloads raw data to a stationary fog over 4G/5G with no analysis, realising \textit{PD}/\textit{OD} tasks and suiting delay-tolerant work such as crop-health logging. \textit{W2} analyses onboard and drops off only the much smaller result, fitting urgent but lightweight models such as irrigation-stress or rodent detection. \textit{W3} hands the raw data to a halting bus over WiFi, which analyses en route and uploads upon entering cellular coverage; it is available where no stationary fog is, and is the only workflow whose completion time is set by a timetable rather than by the UAV. \textit{W4} offloads to a stationary fog that analyses and uploads, the fallback when a model is too heavy for reliable onboard execution. \textit{W2}--\textit{W4} are thus three placements of the analysis step of one \textit{PAD}/\textit{OAD} task; choosing among them per task, jointly with the fleet routes, is the problem we formalize next.

\vspace{-0.05in}
\section{Models and Assumptions}\label{sec:models-and-assumptions}

\begin{figure}[!t]
    \centering
    \includegraphics[width=0.65\linewidth]{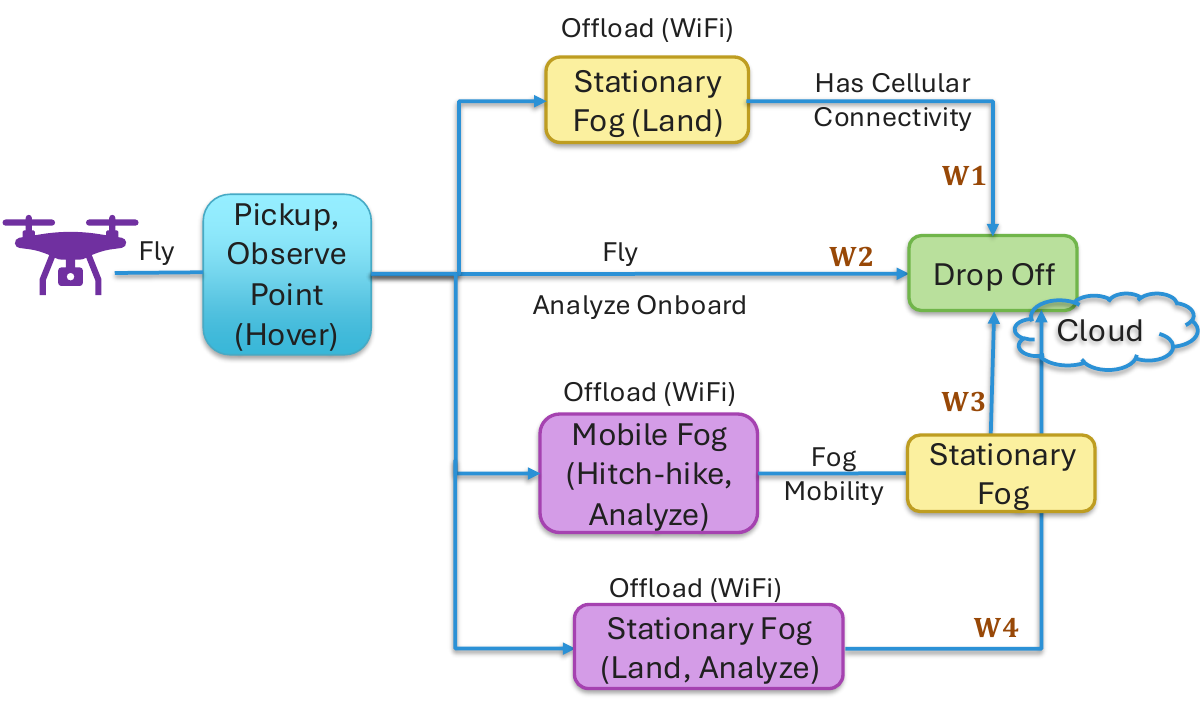}
\caption{Representative Workflows. } 
    \label{fig:workflow-pd-od}
\end{figure}

In this section, we present the system, halt and application models, followed by the latency, energy and cost models for sensing, analytics, communication and mobility. 

\vspace{3pt}\noindent\textbf{Assumptions.} 
UAVs are physically identical, though they may carry heterogeneous compute; only speed, battery and inference rate enter the schedule, so a mixed fleet changes only these numbers. They fly at a constant altitude, and we ignore the mechanics of take-off and landing, which cost the same on every sub-trip. A UAV runs one analysis task at a time, without preemption or checkpointing, although sensing or pickup for one task may overlap the compute of another.

\vspace{-0.05in}
\subsection{System Model}
Let $\mathds{U}$ = $\{u_1, u_2, \dots, u_i\}$ of UAVs managed by a \textit{fleet operator}, initially at a \textit{depot} located at $\Bar{\lambda}$ (latitude, longitude). Each UAV carries a camera and other sensors, GPS, and controls for autonomous flying, an onboard accelerated edge device, and full-duplex WiFi and cellular enabled. 
It hovers at a \textit{waypoint} ($\lambda_n$) to capture observational data or to pick up data from a ground sensor over WiFi. Its battery has a \textit{fixed energy capacity}, $\mathcal{B}$.
Let $\mathds{S}$ = $\{s_1, s_2, \dots, s_j\}$ be the \textit{stationary fogs} co-located at $4G/5G$ cell towers, each $\langle s_j, \hat{\lambda}_j \rangle$, with backhaul to the cloud. A UAV or a bus is \textit{within cellular range} of $s_j$ when inside a radius $R$ of $\hat{\lambda}_j$; we use $R = 1000$~m. One stationary fog serves as the {\em depot}, where the drones start and to which they must return at the end of mission.
Let $\mathds{M}$ = $\{m_1, m_2, \dots, m_k\}$ be the \textit{mobile fogs} carried on public buses. Each $m_k$ runs on a \textit{published timetable route}, halting at bus stops $\hat{\lambda}_k^l$ over windows $(t^h_s, t^h_e)$, represented by $\langle m_k, (\hat{\lambda}_k^1, t^h_s, t^h_e) \dots (\hat{\lambda}_k^l, t^h_s, t^h_e)\rangle$. 
During a halt, a bus \textit{collects data} over WiFi from a UAV perched on its rooftop, and later \textit{relays} it to the cloud once its onward route brings it within cellular range. Data collected by a UAV may thus be \textit{processed onboard}, or \textit{offloaded} to a stationary or mobile fog. The bus also enables UAV battery recharging during the data transfer.

\vspace{-0.05in}
\subsection{Halt and Drop-off Model}\label{sec:halt-model}
The no-detour rule of public transport makes the \textit{halt} the unit of choice. A halt $c = \langle m_k, \hat{\lambda}_k^l, t^h_s, t^h_e \rangle$ is one bus, at one stop, over one time window, and a UAV may hand data to $c$ only if it reaches $\hat{\lambda}_k^l$ before $t^h_e$ and the transfer completes within that window.
What separates one halt from another is not its stop but its bus's onward route. Let $\pi_k(t)$ denote the timetabled position of $m_k$ at time $t$, interpolated along its published route. The \textit{drop-off time} of halt $c$ is the first instant after the handover at which that route enters cellular range:
\begin{equation}\label{eqn:dropoff}
D(c) = \min \{\, t \geq t^h_e \;:\; \exists\, s_j \in \mathds{S},\; \mathcal{H}(\pi_k(t), \hat{\lambda}_j) \leq R \,\}
\end{equation}
The relay opportunity persists until the route leaves range at $\overline{D}(c)$, and the payload transfer must fit within $[D(c), \overline{D}(c)]$. If the remaining route never enters range, $D(c) = \infty$ and the halt is \textit{unusable}. Coverage is evaluated continuously along the route rather than only at stops, so a bus that merely drives past a tower between two stops can still relay. 

\vspace{-0.05in}
\subsection{Application Model}
Let $\tau$ = $\{\tau_1, \tau_2, \dots, \tau_N\}$ denote the set of $N$ \textit{user-defined tasks} to be completed, starting from time $\hat{t} = 0$, where $\tau_n$ 
is a tuple $\langle \psi_n, \lambda_n, t^s_n, \delta^s_n, p^s_n, \mu_n, \phi_n, \delta_n, \beta_n \rangle$: the task type $\psi_n$; the location $\lambda_n = (lat_n, long_n)$ of the ground sensor for \textit{PD/PAD} or of the observation point for \textit{OD/OAD}; the observation duration $t^s_n$ and its deadline $\delta^s_n$; the sensed payload size $p^s_n$ in MB; the DNN model $\mu_n$ used to analyse it and the vector $\phi_n$ of its analysis times on the available compute devices; the task deadline $\delta_n$; and the benefit $\beta_n$ of completing the task within $\delta_n$, reflecting its importance.
A task is {\em completed} if the collected and analyzed data are dropped off to the cloud by the deadline for \textit{PD/OD} and \textit{PAD/OAD} tasks, respectively.
Benefit accrues on \textit{delivery}, not on assignment.
If the total cost incurred by the task $\tau_n$ is $\kappa_n$ (\S\ref{sec:cost-model}), its utility $\gamma_n$ is:
\begin{equation}\label{eqn:utility}
\gamma_n =
\begin{cases}
\beta_n - \kappa_n&\text{task executes within $\delta_n$}, \\
- \kappa_n & \text{task executes but misses deadline},\\
0 & \text{task is dropped (not executed)}
\end{cases}
\end{equation}

\vspace{-0.05in}
\subsection{Latency Model}\label{sec:latency-model}
\vspace{-0.05in}
The \textit{sensing latency} ($t^s$) is the duration for which a UAV hovers at a waypoint ($\lambda_n$) to capture video; it is zero for \textit{PD/PAD} and nonzero for \textit{OD/OAD} tasks. The \textit{analysis latency} ($t^a$) depends on where the analysis runs.
An estimate of analysis latency is provided by a suite of per-device DNN benchmarks, over standard payload sizes, which we enclose in $\phi_i$ for each task.
The \textit{communication latency} ($t^c$) covers (i) $t^c_s = \frac{p^s_n}{R^D_n}$, the \textit{data transfer} from a ground sensor to the UAV over WiFi for \textit{PD/PAD} tasks, and (ii) $t^c_o$, the \textit{offload} from the UAV to a mobile fog over WiFi or a stationary fog via 4G/5G, which is $\frac{p^s_n}{R^U_n}$ for raw data and $\frac{p^o_n}{R^U_n}$ when only the processed result ($p^o_n$) is sent, for uplink and downlink bandwidths $R^U_n, R^D_n$ (Mbps). We assume a stable link of constant bandwidth during offloading.
The \textit{hover-wait latency} ($t^w$) is the time a UAV holds at a bus stop for its bus to pull in, $t^w = \max(0,\, t^h_s - t^{arr})$ for an arrival at $t^{arr}$. It is zero at a stationary fog, and zero if the UAV arrives after the halt has begun. 
When a task is offloaded to a halt $c$, the bus must still carry the payload to coverage. The \textit{relay latency} is $t^r = D(c) - t^h_e$: the interval from the end of the halt until the bus's route first enters cellular range. It is zero when the stop itself lies within $R$ of a tower, and analysis completes before the bus leaves the range, and for a stationary fog.

Let $v^u$ be the average horizontal speed of a UAV and $\mathcal{H}_{x,y}$ the Haversine distance~\cite{6815214} between locations $\lambda_x$ and $\lambda_y$. The \textit{mobility latency} for a task, from data pickup at $\lambda_x$ to a drop-off at a stationary fog $\lambda_z$ or a halt $\lambda_y$, is:
\begin{equation}\label{eqn:mobility-time}
t^m =
\begin{cases}
\frac{\max(0,\; \mathcal{H}_{x,z} - R)}{v^u}, & \text{stationary fog} \\
\frac{\mathcal{H}_{x,y}}{v^u} + t^w + t^r, & \text{mobile fog}
\end{cases}
\end{equation}
The UAV--tower link is cellular, so a UAV offloads on entering the footprint and flies only the excess beyond $R$. However, for the mobile fog, the drone needs to fly all the way to the halt, as a WiFi handover needs proximity.
The payload then reaches the cloud at $D(c)$, fixed by the timetable and independent of the UAV.
The \textit{overall latency} for a task is:
    $\mathds{T} = t^s + t^a + t^c + t^m$.

\vspace{-0.05in}
\subsection{Energy Model} \label{sec:energy-model}
The UAV draws on its battery capacity $\mathcal{B}$ across the operations of a task.
The \textit{sensing energy ($e^s$)} is consumed by the onboard camera only for \textit{OD} and \textit{OAD} tasks.
The \textit{analysis energy} ($e^a$) is consumed for analysis onboard the UAV edge accelerator and given by $e^a = t^a \times \mathds{V} \times \mathds{P}_a$, where $\mathds{V}$ is the CPU/GPU utilization and $\mathds{P}_a$ is the rated power of the edge resource.
The \textit{communication energy} ($e^c$) consumed during data transfer is given by $e^c = t^c \times \mathds{P}_c$, where $\mathds{P}_c$ is the power consumed during communication. The \textit{mobility energy} $e^m = e^f + e^h$ covers flight to a waypoint ($e^f$) and hovering to sense, pick up, or wait out $t^w$ at a stop ($e^h$), from the actual distances and durations; it is shared among all tasks active on the drone while it flies.
So, the \textit{cumulative energy consumption} for a task is:
    $\mathds{E} = e^s + e^a + e^c + e^m$.

UAVs may recharge their batteries once the residual charge falls below a fraction $\theta$ of $\mathcal{B}$. Every fog is a recharge point: a stationary fog and a bus halt. 
Separately, a drone whose residual charge falls below a reserve $\rho\mathcal{B}$, or that cannot afford its next hop, is diverted to the nearest reachable recharge point.

\vspace{-0.05in}
\subsection{Cost Model} \label{sec:cost-model}
The fleet operator charges the user per task. The \textit{sensing cost} $c^s = c^s_{\bullet} \times t^s$ applies to \textit{OD/OAD} tasks and is zero for \textit{PD/PAD}. The \textit{analysis cost} is $c^a = c^a_{\bullet} \times t^a$, at the per-unit-time rate of whichever edge or fog device runs it. The \textit{communication cost} is $c^c = c^c_{\bullet} \times t^c$, for a per-MB transfer rate $c^c_{\bullet}$; it is zero for a handover to a bus, which is over WiFi. The \textit{mobility cost} is $c^m = c^e_{\bullet} \times e^m$, at a pro-rata rate per unit of UAV energy, covering flight and all hovering including $t^w$. The \textit{total cost} of a task is $\kappa = c^s + c^a + c^c + c^m$; we ignore fixed costs of operating the captive fleet. Rates come from 
provider price lists, or literature~\cite{raj2025adaptive}.

\vspace{-0.05in}
\section{Mission Scheduling Problem Formulation}\label{sec:problem-formulation} 

A UAV fleet operator \textit{queues} incoming tasks and periodically plans a \textit{mission schedule} that serves some or all of them with the whole fleet, so as to maximise utility.

A \textit{drone trip} begins and ends at the depot and comprises several \textit{sub-trips}, each running from one fog to another through a sequence of task waypoints. The route of the $r^{th}$ sub-trip of drone $u_i$ is $R^r_i = (\hat{\lambda}_p, \lambda^r_{i_1},\dots,\lambda^r_{i_{n}},\hat{\lambda}_q)$, with each intermediate $\lambda^r_{i_n}$ the waypoint of a task $\tau_n$. The first sub-trip departs the depot and the last returns to it, and consecutive sub-trips are contiguous, $R^r_i[-1] = R^{r+1}_i[0]$. A drone ending a sub-trip at any fog, a stationary point, or a bus halt, may recharge its battery if its residual charge is $< \theta\mathcal{B}$ (\S~\ref{sec:energy-model}).

Let $T_i^r \subseteq \tau$ be the ordered tasks of that sub-trip, where $\Bar{t}$ and $\Bar{\Bar{t}}$ bound a pickup or observation. Their sensing intervals do not overlap, $\Bar{\Bar{t}}_{i_n}^r \leq \Bar{t}_{i_{n+1}}^r$; each task is served by at most one drone, $T^r_i \cap T^{r'}_{i'} = \varnothing$; and infeasible/unprofitable tasks are dropped, so $\sum_i \sum_r \mid{T^r_i}\mid \leq |\tau|$. With $\mathcal{F}(\cdot,\cdot)$ the mobility time of Eq.~\ref{eqn:mobility-time}, a drone leaving $\lambda^r_{i_n}$ at $d^r_{i_n}$ reaches next waypoint at $a^r_{i_{n+1}} = d^r_{i_n} + \mathcal{F}(\lambda^r_{i_n},\lambda^r_{i_{n+1}})$, hovers to capture data, and must arrive by $\Bar{t}^r_{i_n} = \delta^{s}_{n} - t^{s}_{n}$ or the sensing deadline is already lost.

Further, the data handover to a mobile fog must fit inside the halt. Writing $\hat{a} = \max(a^r_{ic},\, t^h_s(c))$ for when the transfer can begin at a halt $c$, a drone arriving early hovers for $t^w$ until the bus pulls in, we require $\hat{a} + t^c_o \leq t^h_e(c)$; a drone reaching the stop after $t^h_e(c)$ has missed the bus and wasted the flight. Finally, a task completes only when its payload reaches the \textit{cloud} within $\delta_n$, not when it is handed over. At a stationary fog $s_j$ the drone need only enter the tower's range, giving $a^r_{ij} \leq \delta_n - t^a_{j_n} - t^c_o$; at a halt, delivery waits on the bus, so any analysis aboard must finish by $D(c)$ and $D(c) \leq \delta_n$.
Since $D(c)$ is fixed by the timetable rather than by the flight, this is the one deadline the scheduler cannot buy down, and it is what makes a farther halt with an imminent relay preferable to a nearer one without.

For a task needing analysis, $o^{r}_{i_n} \in \{0,1\}$ selects onboard execution over offloading; an onboard task occupies a slot $[\theta_{i_n}, \theta_{i_n} + t^a_{i_n})$ that may begin only after capture and may not overlap another on the same drone. The trip ends at the depot once the last task of the last sub-trip has been dropped off. Each task earns $\gamma_n$ by Eq.~\ref{eqn:utility}, and MSP maximises the \textit{total utility} $\gamma = \sum_{n} \gamma_n$, subject to the above and to the per-sub-trip energy budget $\mathcal{B}$.

\subsection{NP-Hardness of MSP}
\vspace{-0.03in}
The \textit{optimization objective} of MSP is given by:
\begin{equation}\label{eq:msp_objective}
    \arg \max_{R} \sum_{\tau_n\in \tau} \gamma_n \notag
\end{equation}
It needs to assign UAVs to task waypoints and decide whether to compute onboard or offload to fogs to maximize the utility from task completion, subject to the constraints:
$\mathds{T}_{n} < \Bar{\delta}_{n}$ and 
$\mathds{E}_{n} < \mathcal{B}$.
These two constraints imply tasks must be completed \textit{before their effective deadline}, and a sub-trip must end \textit{before the drone battery drains}. 

\begin{theorem}
MSP is NP-Hard.
\end{theorem}
\vspace{-0.15in}

\begin{proof}
We reduce the Vehicle Routing Problem (VRP), known to be NP-hard~\cite{lenstra1981complexity}, to MSP. Given a VRP instance over customers and a depot, construct in polynomial time an MSP instance with one \textit{PD} task of negligible payload per customer, relaxed time windows ($\Bar{t} = 0$, $\Bar{\Bar{t}} = +\infty$), one sub-trip per drone, the depot as the only fog, and $\mathcal{B}$ set to the vehicle capacity. Sensing, analysis and communication costs vanish, so each $\kappa_n$ reduces to the mobility cost and maximising $\sum_n \gamma_n$ is exactly minimising the total routing cost; an optimal MSP schedule thus yields an optimal VRP solution and vice versa. Since the general MSP further admits time windows, multiple sub-trips and a choice of fogs, MSP is NP-hard.
\end{proof}

\vspace{-0.1in}
\section{Divide and Assign (DA) Heuristic for MSP}\label{sec:heuristics}
Since MSP is NP-hard, we schedule greedily rather than exactly. We propose \textit{Divide and Assign} (DA), which partitions tasks by location and deadline, then assigns each to the cheapest UAV that can still deliver it, choosing a drop-off point as it goes. It has two steps and a pre-processing stage.

\textbf{Step 0 (Pre-process):} Two structures are built once per mission. Over the stationary fogs, a \textit{K-D tree} answers nearest-tower queries in logarithmic time on average~\cite{friedman1977algorithm}. The drop-off time $D(c)$ and relay window of every halt are computed from  published timetables and tower positions (\S\ref{sec:halt-model}), so no route geometry is recomputed during assignment.

\setlength{\textfloatsep}{0pt}
\begin{algorithm}[t!]
    \caption{\textbf{Divide and Assign (DA)}}
    \label{algo:DA}
    \small
    \DontPrintSemicolon
    \KwIn{$\mathds{S}, \mathds{M}, \mathds{U}, \tau$}
    \KwOut{Assigned tasks for each UAV $u_i$}

    \tcc{Step 0: Pre-processing}
    Build a K-D tree over the stationary fogs $\mathds{S}$.\;
    Precompute $D(c)$ and the relay window of every halt $c \in \mathds{M}$.\;

    \tcc{Step 1: Spatio-temporal Clustering}
    $\mathbb{G} \gets STDBSCAN(\tau)$\;
    Sort tasks $\tau_n \in g$ in ascending order of their deadlines $\delta_n$.\;
    Sort groups $g \in \mathbb{G}$ by their earliest deadline.\;

    \tcc{Step 2: Task to UAV assignment}
    \For{each $g \in \mathbb{G}$}{
        \For{$\tau_n \in g$}{
            Initialize the assigned drone for the task to $\varnothing$\;
            Initialize min cost $\kappa_{min}$ to $\infty$\;
            \For{$u_i \in \mathds{U}$}{
                $\kappa_{i,n} \gets c_i^s + c_i^a + c_i^c + c_i^m$\;
                \If{ $\kappa_{i,n} < \kappa_{min}$}{
                    $f \gets$ cheapest admissible halt in $\mathds{M}$, else the nearest stationary fog\;
                    \If{$f = \varnothing$, or $u_i$ misses a deadline or lacks charge for $\tau_n$}{
                        Divert $u_i$ to the nearest reachable fog to drop off, and  recharge battery if below $\theta\mathcal{B}$.\;
                    }
                    \Else{
                        Update the assigned drone to $u_i$.\;
                    }
                }
            }
            \If{Onboard analysis cost is lower than at $f$ and onboard analysis satisfies both time and energy constraints}{
                Schedule $\tau_n$ for analysis at the assigned drone.\;
            }
            \Else{
                Analyze $\tau_n$ at the drop-off fog $f$ of the drone's subtrip.\;
            }

        }
    }
    For all drones drop any pending payload at their chosen fog and return to the depot.\;
\end{algorithm}

\textbf{Step 1 (Divide):} Close tasks may carry very different deadlines, while tasks sharing a deadline may be far apart, so we cluster on both axes with \textit{ST-DBSCAN}~\cite{birant2007st}. Within each cluster $g \in \mathbb{G}$ tasks are sorted by deadline, and clusters by their earliest deadline, $\min(\{\delta_n : \tau_n \in g\})$.

\textbf{Step 2 (Assign):} Taking clusters in that order, each task goes to the UAV of least total cost $\kappa_{i,n} = c^s_i + c^a_i + c^c_i + c^m_i$ among those that can reach the waypoint before its sensing deadline, deliver the payload by $\delta_n$, and do so within residual charge. An analysis task runs onboard when that is cheaper than fog compute and both time and energy allow; otherwise it is analysed at the fog that ends the sub-trip.

\textbf{Choosing the drop-off:} This is where the halt model enters, and it is what separates DA from a proximity-driven scheduler. A halt $c$ is \textit{admissible} only if the UAV reaches the stop before the bus departs, the handover fits inside the window, the onward route does enter coverage, and the delivery it implies both meets the deadline, $D(c) \leq \delta_n$, and falls inside the relay window $[D(c), \overline{D}(c)]$. DA takes the cheapest admissible halt by flight and hover wait, $c^e_{\bullet}(e^f(c) + e^h(t^w(c)))$, with $D(c)$ breaking ties, since the WiFi handover is free and on-bus compute is halt-independent. It keeps whichever is cheaper, that halt or the nearest stationary fog. Ranking by cost rather than proximity is what lets a farther stop with an imminent relay beat a nearer one whose bus stays out of coverage.

\textbf{Battery:} A drone below the reserve $\rho\mathcal{B}$, or unable to afford its next hop, is diverted to the nearest reachable fog to drop off, and recharges there if below $\theta\mathcal{B}$ (\S\ref{sec:energy-model}). A recharge needs contact, whereas offloading to a tower needs only cellular range, so a drone needing a charge flies the full distance, whereas one merely dropping off does not.

\vspace{-0.05in}
\subsection{Time Complexity}
\vspace{-0.05in}
Building the stationary-fog K-D tree takes $\mathcal{O}(|\mathds{S}| \log|\mathds{S}|)$ time, and precomputing $D(c)$ over all halts is linear in $|\mathds{M}|$ given the timetables. ST-DBSCAN partitions the tasks in $\mathcal{O}(|\tau| \log|\tau|)$ and fixes the group count $|\mathds{G}|$ automatically~\cite{birant2007st}; the sorting within and across groups adds $\mathcal{O}(|\tau| \log\frac{|\tau|}{|\mathds{G}|} + |\mathds{G}| \log|\mathds{G}|)$.

In Step~2, each of the $|\tau|$ tasks is tried against $|\mathds{U}|$ drones. A nearest-tower query costs $\mathcal{O}(\log|\mathds{S}|)$, but the halts must be scanned linearly at $\mathcal{O}(|\mathds{M}|)$: admissibility turns on timing rather than position, so no spatial index can prune it. Step~2 therefore takes $\mathcal{O}(|\tau| |\mathds{U}| (|\mathds{M}| + \log|\mathds{S}|))$, and Algorithm~\ref{algo:DA} takes
$\mathcal{O}(|\tau| |\mathds{U}| (|\mathds{M}| + \log|\mathds{S}|) + |\tau| \log|\tau| + |\mathds{S}| \log|\mathds{S}|)$ time.
The linear scan is the price of the no-detour model, and it is cheap in practice, as our regions carry at most $100$ halts. The clustering cost is incurred offline during periodic mission planning, so it does not affect real-time flight operations.

\vspace{-0.05in}
\subsection{Adapting to Traffic}
\vspace{-0.05in}

A rural bus is often late, and a halt that was admissible when the sub-trip was planned can therefore be infeasible by the time the UAV arrives. Since the only action left is to hover, at $700$~J/s compared to $750$~J/s for flight, a UAV that waits too long leaves its payload and charge stranded.
\textit{Adaptive DA (ADA)} extends DA with a real-time transit feed. Planning is unchanged; what changes is that a delay observed at execution shifts $t^h_s$, $t^h_e$, and $D(c)$ of the affected halt together, and DA's own admissibility test is re-evaluated against the shifted values. Both variants we implement poll the feed at the same instant, $\tau$ seconds before the UAV would reach the halt it is committed to. They differ in \textit{how long the correction survives}.

\textit{ADA-Local} uses what it learns once and then discards it. If the delayed halt is still admissible, the UAV proceeds and simply hovers longer; if not, it re-runs the drop-off rule over every other admissible halt and the nearest stationary fog, from its position along the leg and with the flight energy already spent deducted. Routing and task assignment are untouched, and the corrected times are not carried forward, so the next sub-trip is planned against the published timetable again and can commit to the same late bus. The repair is $\mathcal{O}(|\mathds{M}|)$ per affected sub-trip and adds no additional planning cost.
On the contrary, \textit{ADA-Global} keeps it. The same poll writes the corrected times into the fleet's shared view of the timetable, where they persist, so every later decision works from them: not only subsequent drop-off choices, but the admissibility test DA applies when it \textit{assigns} a task in the first place, which ADA-Local never reaches. A UAV, therefore, stops committing to halt the fleet that is already known to be late, instead of repeatedly discovering them en route.

\section{Experimental Study} \label{sec:evaluation}
\subsection{Implementation and Experimental Setup}
\begin{table}[t]
\setlength{\tabcolsep}{2.5pt}
\caption{Edge and mobile fog computing devices used in evaluations. }
\label{tab:computespecs}
\centering
\begin{tabular}{l|r|r|r}
\hline
\textbf{Feature} & \bf {Orin AGX} & \bf {Xavier NX } &\bf{Orin Nano} \\
\hline
\hline
GPU Architecture & Ampere & Volta & Ampere\\
\noalign{\global\arrayrulewidth=0.1pt}\arrayrulecolor{lightgray}\hline
\noalign{\global\arrayrulewidth=0.4pt}\arrayrulecolor{black}
\# CUDA/Tensor Cores & 2048/64 & 384/48 & 1024/32\\
\noalign{\global\arrayrulewidth=0.1pt}\arrayrulecolor{lightgray}\hline
\noalign{\global\arrayrulewidth=0.4pt}\arrayrulecolor{black}

RAM (GB) & $32$ & $8$ & $8$ \\ \hline
\noalign{\global\arrayrulewidth=0.1pt}\arrayrulecolor{lightgray}\hline
\noalign{\global\arrayrulewidth=0.4pt}\arrayrulecolor{black}
Energy Usage (J/s) & 60 & 15 & 15\\ \hline
Form factor (mm) & $110 \times 110 \times 72$ & $103 \times 90 \times 35$ & $100 \times 79 \times 21 $\\
\noalign{\global\arrayrulewidth=0.1pt}\arrayrulecolor{lightgray}\hline
\noalign{\global\arrayrulewidth=0.4pt}\arrayrulecolor{black}
Weight (g) & 872.5 & 174 & 176\\ \hline
Price (INR) & $2,05,000 (\$2360) $  & $40,000 (\$460)$ & $54,500 (\$630)$\\
\noalign{\global\arrayrulewidth=0.1pt}\arrayrulecolor{lightgray}\hline
\noalign{\global\arrayrulewidth=0.4pt}\arrayrulecolor{black}
Unit Cost (INR/s) & $2.16e^{-3}$  & $4.2e^{-4}$ & $5.7e^{-4}$\\
\hline
\end{tabular}
\end{table}

\begin{table}[t]
\small
\setlength{\tabcolsep}{1.5pt}
\caption{Unit cost (INR/s), energy (J/s).}
\label{tab:workloads}
\centering
\begin{tabular}{l|r|r|r|r}
\hline
\textbf{Parameter} & \bf {Sense} & \bf {Analysis\tablefootnote{
Orin Nano, Xavier NX, Orin AGX, RTX 3090. Energy use is not applicable (N/A) for fogs as they are not powered by the drone's battery.}} &\makecell{\bf{Comms.}\\\textit{(WiFi, Cellular)}}
&\makecell{\bf{Mobility}\\
\textit{(Fly, Hover)}}\\
\hline\hline
Unit Cost & $1.3e{-4}$ & \makecell{$5.7e{-4}, 4.2e{-4},$ \\$21.6e{-4}, 40e{-4}$} & 0, $1.16e{-8}$ & \makecell{$12.5e{-4},$\\ $11.6e{-4}$} \\
\hline
Energy Use & 0.7 & 15, 15, N/A, N/A & 8.26, 37.5 & 750, 700\\
\hline
\end{tabular}
\end{table}

We implement all schedulers in Python $3.12$ and run them on an Apple M4 Max with $36$~GB of unified memory, under macOS $26$. Each scheduling run is a single-threaded process, and the harness executes up to $10$ such runs concurrently.
We benchmark flying and hovering on a \textit{commercial-grade quadcopter} with an onboard camera~\cite{eCAM80_CUNX} drawing $700$mW and either a Jetson Orin Nano~\cite{orin-tech-specs} or Xavier NX~\cite{nvidia_xavier_nx} as its edge accelerator. Mobile fogs carry a Jetson Orin AGX~\cite{nvidia_orin_agx} and stationary fogs a GeForce RTX 3090~\cite{geforce_rtx_3090} with $10{,}496$ CUDA cores and $24$GB GDDR6X; Table~\ref{tab:computespecs} gives their specifications and Table~\ref{tab:workloads} the unit costs in INR/s, derived from~\cite{raj2024adaptive}.
Drones fly at $v^u = 4$~m/s and buses travel at $v^m = 5.14$~m/s ($\approx 18km/hr$), with a drone battery capacity of $\mathcal{B} = 1350$~kJ based on popular UAVs. The cellular range is $R = 1000$~m for both the UAV--tower link and the bus relay.
Drones enable battery recharge below $\theta = 0.30$ of capacity and divert below a reserve of $\rho = 0.15$. Every configuration is evaluated on $10$ seeded workloads, with the fog and bus infrastructure held fixed so that the workload is the only source of variation.

\subsection{Workloads}\label{sec:workloads}
We benchmark four \textit{agriculture-relevant DNN models} on the \textit{Rice Paddy Dataset}~\cite{petchiammal2023paddy}: MobileNet-v3 (object classification in farms), YOLOv8-m (object detection), ExpansionNetv2 (image captioning) and a re-trained ResNet-34 (disease and pest detection). Each runs on payloads of $0.8$, $8$, $80$, $800$ and $4000$~Mb across the four GPUs. Table~\ref{tab:inference-times} reports the $95^{th}$ percentile inference time per image, over $20$ runs at $4000$~Mb and $1{,}000$ runs otherwise; these populate $\phi_n$. Entries marked $\infty$ exceed onboard memory.
Stationary fogs are placed from $4$G cell-tower data extracted from OpenCelliD~\cite{OpenCelliD}.
Halts come from the published rural bus timetables of the region under study, as a GTFS feed~\cite{BMTC_Official}. Each bus revisits its stops on its schedule and relays through whichever tower its route next reaches.
We use the $5\times5km^2$ region of farmland area for our evaluation, shown in Fig.~\ref{fig:map-5x5}. Every event (trip, stop) in the feed falling inside the box becomes a halt, timed by the published schedule and routed along the feed's road geometry. 
The region we report our results on holds a single stationary fog, $24$ stops, and $160$ halts over $32$ trips, and its cellular coverage reaches only $12.5\%$ of the area. This is the sparse regime the transit tier is meant for. We also build a $10\times10km^2$ box around the same area for an extended study, with $3$ stationary fogs, $81$ stops and $407$ halts over $71$ trips at $8.3\%$ coverage.

\begin{figure}[t]
    \centering
    \subfloat[$5\times5km^2$]{
    \includegraphics[width=0.35\columnwidth]{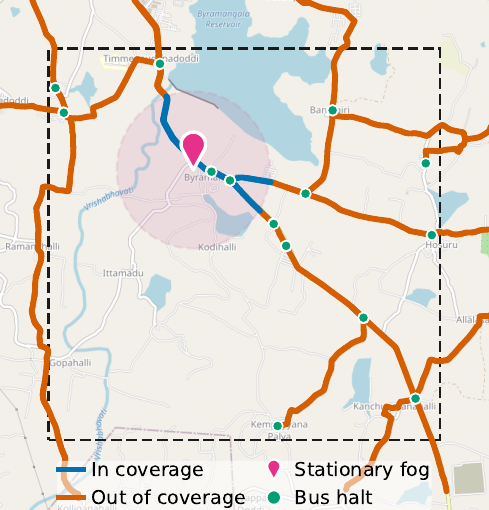}
    \label{fig:map-5x5}
    }~
    \subfloat[$10\times10km^2$]{
    \includegraphics[width=0.35\columnwidth]{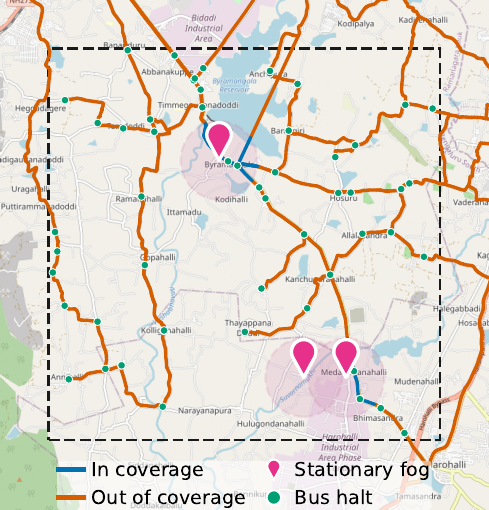}
    \label{fig:map-10x10}
    }  
   \caption{Our study regions on an OSM base map. Pins are the stationary fogs, shaded discs their cellular range, and lines are the bus routes the mobile fogs follow.}
    \label{fig:region-headtohead}
\end{figure}

\begin{table}[t]
    \caption{$95^{th}$ percentile inference time (s) per image frame, by DNN model, GPU device and payload size. $\infty$ marks a model that does not fit on the device.}
    \label{tab:inference-times}
    \centering
    \footnotesize
    \setlength{\tabcolsep}{3pt}
    \begin{tabular}{c c c c c c c}
        \toprule
        & & \multicolumn{5}{c}{ \textbf{Data Payload Sizes (in Mb)} } \\
        \textbf{Device}& \textbf{DNN Model} & \textbf{0.8} & \textbf{8} & \textbf{80} & \textbf{800} & \textbf{4000} \\
        \midrule
     & YOLOv8(m) & 0.24 & 2.22 & 22.40 & 98.36 & $\infty$ \\
        Orin Nano & ExpansionNetv2 & 1.01 & $\infty$ & $\infty$ & $\infty$ & $\infty$ \\
         & MobileNet-v3 & 0.18 & 0.43 & 1.69 & 19.38 & 89.21 \\
         & ResNet-34 & 0.18 & 0.16 & 1.76 & 20.47 & 458.35 \\
         \midrule
         & YOLOv8(m) & 0.30 & 3.11 & 30.53 & 223.45 & $\infty$ \\
       Xavier NX  & ExpansionNetv2 & 1.07 & $\infty$ & $\infty$ & $\infty$ & $\infty$ \\
         & MobileNet-v3 & 0.19 & 1.60 & 2.31 & 23.16 & 107.84 \\
         & ResNet-34 & 0.21 & 1.48 & 2.88 & 29.51 & 390.27 \\
         \midrule
         & YOLOv8(m) & 0.07 & 0.30 & 3.25 & 30.63 & 207.10 \\
        Orin AGX & ExpansionNetv2 & 0.37 & 1.68 & 20.48 & 208.77 & 1071.80 \\
         & MobileNet-v3 & 0.04 & 0.21 & 0.84 & 9.89 & 48.02 \\
         & ResNet-34 & 0.04 & 0.17 & 0.96 & 8.92 & 45.64 \\
         \midrule
         & YOLOv8(m) & 0.02 & 0.12 & 1.47 & 11.44 & 79.02 \\
        RTX 3090 & ExpansionNetv2 & 0.09 & 0.31 & 3.81 & 37.37 & 191.59 \\
         & MobileNet-v3 & 0.01 & 0.04 & 0.40 & 3.90 & 20.19 \\
         & ResNet-34 & 0.01 & 0.04 & 0.42 & 4.16 & 21.70 \\
        \bottomrule
    \end{tabular}
\end{table}

We use the task type ($PD$, $OD$, $PAD$, $OAD$) to evaluate for nine workload cells: $\{10, 20, 50\}$ drones at load factors $x = n/i \in \{2,4,8\}$, i.e., from $20$ tasks on $10$ drones up to $400$ on $50$. A drone takes up to $r_{max} = x$ sub-trips. Waypoints are sampled uniformly at random from the region, and tasks may start at any time within $(0, 240]$~mins. Table~\ref{tab:taskparams} gives the sampled task parameters.

\begin{table}[t]
\small
\centering
\caption{\bf\footnotesize{Sampled task parameters.}}
\label{tab:taskparams}
\vspace{-0.05in}
\begin{tabular}{c|c|c|c|c}
\toprule
$p^s$ & $\delta^s, \delta$ & $\delta - \delta^s$ & $\beta$  & $x$\\
\midrule
$[0.1,500]$MB& $[20,240]$min & $\geq5$min & $[3,5]$ & $[2,4,8]$\\
\bottomrule
\end{tabular}
\end{table}

\begin{figure*}[!t]
\centering
    \subfloat[Pickup Drop (PD)]{%
\includegraphics[width=0.45\columnwidth]{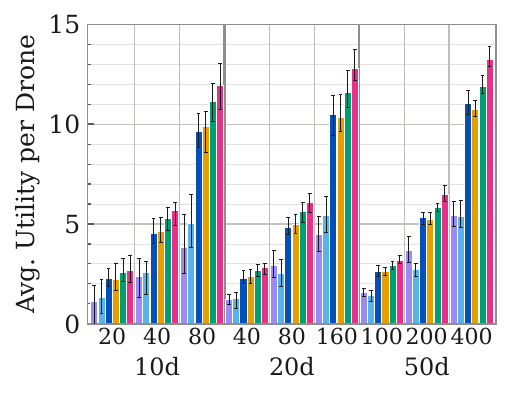}%
   \label{fig:pd11_blr_utility}%
  }%
  \subfloat[Observe Drop (OD)]{%
\includegraphics[width=0.45\columnwidth]{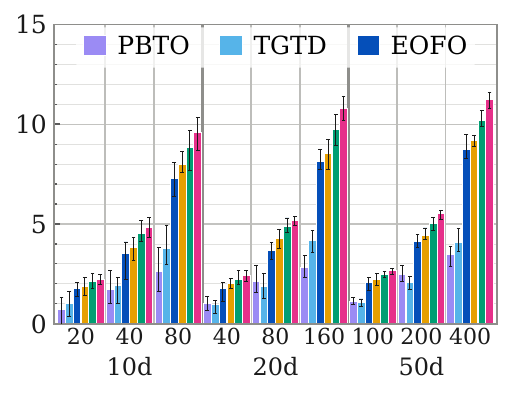}%
   \label{fig:od11_blr_utility}%
  }\\
  \subfloat[Pickup Analyse Drop (PAD)]{%
    \includegraphics[width=0.45\columnwidth]{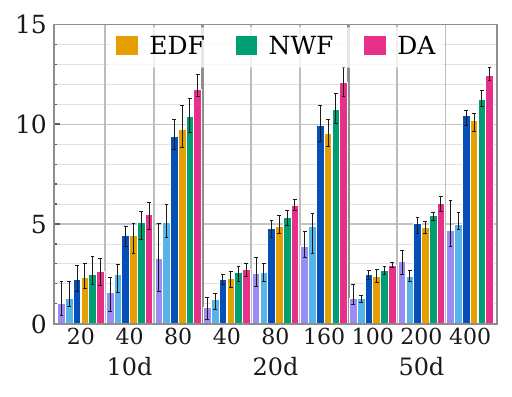}%
   \label{fig:pad11_blr_utility}%
  }%
  \subfloat[Observe Analyse Drop (OAD)]{%
\includegraphics[width=0.45\columnwidth]{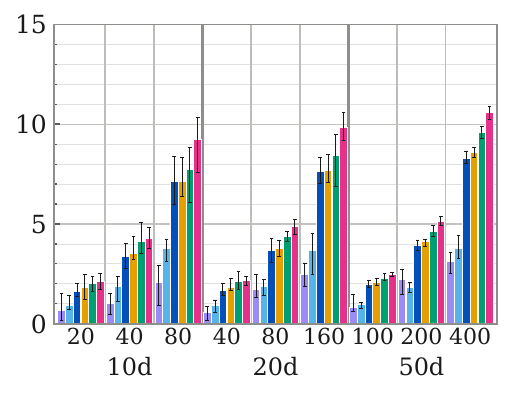}%
   \label{fig:oad11_blr_utility}%
  }
\caption{Comparison of algorithms across different task types 
for 10 seeded iterations.
}
\label{fig:algo-comparison}
\end{figure*}

\subsection{Baseline Schedulers}\label{sec:baselines}
We first compare DA against two simple \textit{heuristics}: \textit{Earliest Deadline First (EDF)}~\cite{stankovic1998deadline} and \textit{Nearest Waypoint First (NWF)}~\cite{perez2013solving}, which sort tasks by ascending deadline and by distance from the depot respectively, then give each task to the \textit{first} UAV satisfying the energy and deadline constraints, where DA scans every UAV and keeps the cheapest. Everything else is DA's intelligence.
The baselines, therefore, differ from DA only in task ordering and the first-fit UAV choice, thereby isolating those two effects.

We further compare DA against three \textit{SOTA} schedulers, each ported from their respective papers' decision rules and then scored using \textit{our} utility, cost, and deadline accounting.
Park et al.'s successive convex approximation (SCA) scheme~\cite{park2024public}, which we label \textit{PBTO}, is the only prior work that offloads UAV computation to public buses. We port its bus selection: candidate buses are admitted incrementally until the reduction in system cost falls below a threshold (their Alg.~1), and each is scored 
over task delay and UAV energy (their Eq.~18, Table~II).
We change two things. Their UAV stays in one place, so we use DA's routing. 
A stationary fog is used only when no halt meets the deadline, as in their own edge-server variant.
\textit{Earliest Observation First with On-time processing (EOFO)}~\cite{10146342}, co-schedules routes with onboard analytics. It sorts activities by non-descending earliest observation start time and, among the drones on which an activity is feasible, takes the one of \textit{lowest incremental energy} rather than lowest cost.
\textit{TGTD}~\cite{li2026latency} is a two-stage framework that combines heuristic search with deep reinforcement learning to jointly plan a UAV's visiting order and its offloading. We port its first stage, a heuristic search over task orderings that favors shorter routes, stay near a fog where the UAV can recharge, and reach urgent tasks early; its second stage sends each task to whichever finishes sooner: the UAV's accelerator or the nearest stationary fog. Their offload fraction $\alpha$ becomes all-or-nothing, because a task is a single inference.

\subsection{Comparison with Baseline Schedulers}\label{sec:analysis-1}

\begin{figure}[t]
    \centering
    \includegraphics[width=0.5\linewidth]{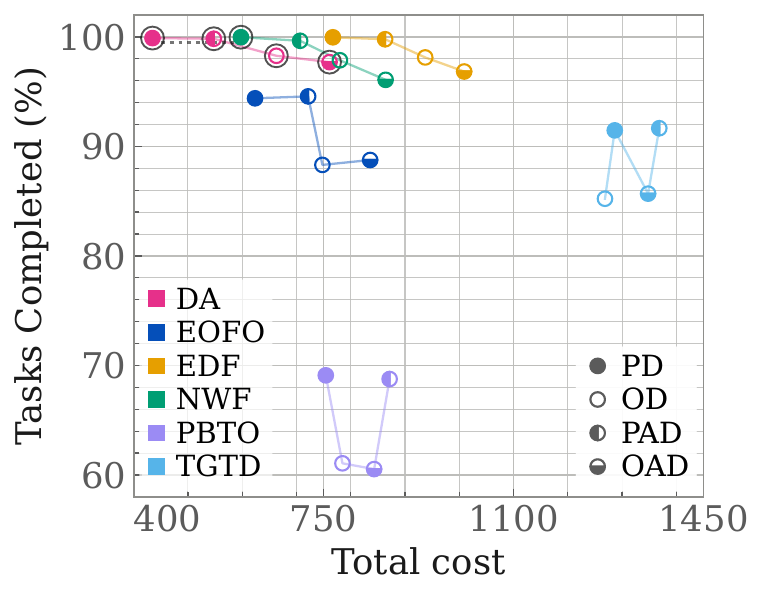}
    \caption{Cost-completion Pareto view for $5\times5km^2$ region, average of $10$ seeded instances, summed over $9$ workloads.}
    \label{fig:pareto-cost}
\end{figure}

Fig.~\ref{fig:algo-comparison} reports the average utility per drone for all six schedulers across the four task types and nine workloads, with each bar being the mean of ten seeded replicates, and whiskers spanning the minimum and maximum.
As we observe, DA achieves the highest utility across all $36$ bars, up to $20\%$ higher than the strongest baseline (NWF). Moreover, when measured against the strongest SOTA scheduler (EOFO) in each cell, the lead grows up to $41\%$, with EOFO attaining $18$--$24\%$ less utility than DA on aggregate. EOFO drops off only at the depot, so that margin reflects DA's access to the transit tier as well as its strategy. The ordering behind it is stable: NWF is consistently the closest competitor ($6.22$ against DA's $6.89$ for PAD), EOFO and EDF sit a further step back and change places between task types, and the two SOTA ports are less than half of DA. 
The margin over NWF increases with load factor; i.e., for PAD, the gap over NWF widens from $5.2\%$ at $20T10d$ to $10.9\%$ at $400T50d$, which makes DA scalable.

The two SOTA ports separate for reasons specific to their decision rules rather than to tuning. PBTO ranks candidate halts on a nearest-first basis, so it is possible that it will commit a UAV to a stop whose bus does not re-enter the cellular range before the task deadline, Data are handed over but never dropped off, and the task is lost. TGTD has no mobile fog tier at all, so in a region where $87.5\%$ of the area is out of tower range, it must route every delivery back to one stationary fog.

A higher utility could, in principle, be achieved by spending more or by quietly abandoning the hard tasks, so Fig.~\ref{fig:pareto-cost} plots total cost against on-time completion for the same runs. Cost is lower-better and completion higher-better, so the best corner is top-left, and the ringed points joined by the dotted staircase are the \textit{Pareto-optimal} ones, those that no other scheduler beats on both axes at once. DA is Pareto-optimal in every task type, and in OD, PAD, and OAD it is the \textit{only} non-dominated point, so no baseline is at once cheaper and more complete; PD is the single task type it shares, with NWF. We show that DA is the cheapest scheduler across all four task types and is within $0.1\%$ of the best completion.
Task completion is also what most sharply separates the SOTA ports, with PBTO finishing only $61$--$69\%$ of tasks on time compared to DA's $98$--$100\%$.

\subsection{Value of Each Infrastructure Tier}\label{sec:analysis-ablation}

\begin{figure}[t]
    \centering
    \includegraphics[width=0.7\linewidth]{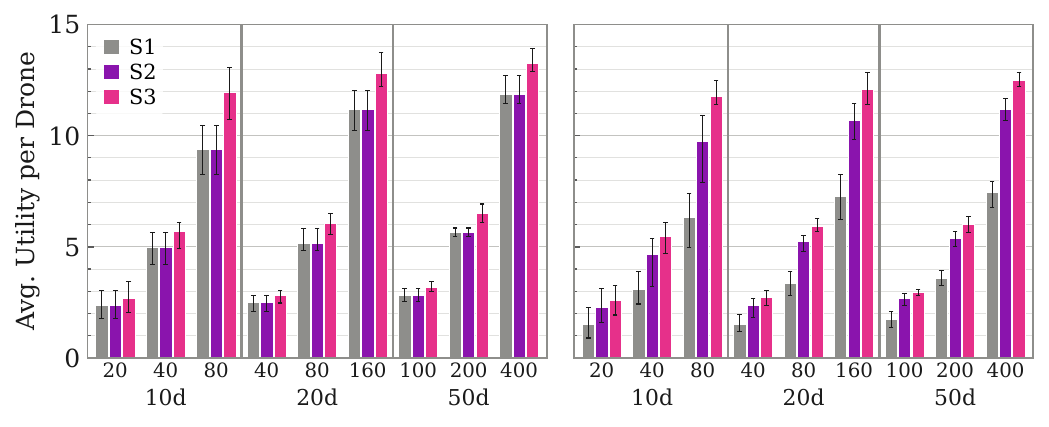}
\caption{Infrastructure ablation of DA 
    for PD (left) and PAD (right). 
    }
    \label{fig:da-ablation}
\end{figure}

\begin{figure}
    \centering
    \includegraphics[width=0.5\linewidth]{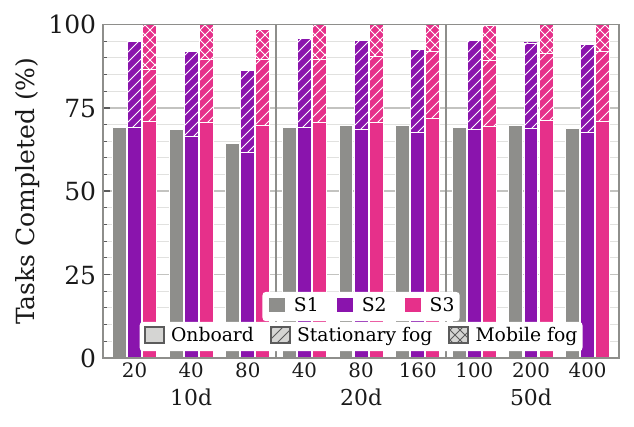}
    \caption{Task completion for the same ablation and region as Fig.~\ref{fig:da-ablation} for PAD. Hatched segments partition it into analysis performed onboard, at a stationary fog, or at a mobile fog.}
    \label{fig:da-ablation-analysis}
\end{figure}

\begin{figure}[t]
    \centering
    \includegraphics[width=0.7\linewidth]{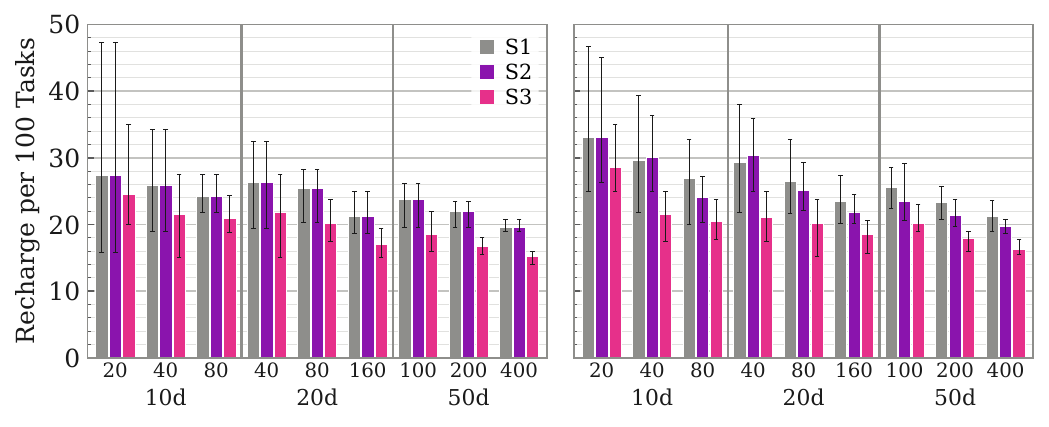}
    \caption{Recharging cycles per $100$ completed tasks for the same ablation and region as Fig.~\ref{fig:da-ablation}, for PD (left) and PAD (right).}
    \label{fig:da-ablation-swaps}
\end{figure}

The comparisons so far hold the infrastructure fixed and vary the scheduler. Fig.~\ref{fig:da-ablation} does the reverse, running DA over three variants of the environment that each add one capability: S1 runs every analysis onboard, S2 switches the stationary fog compute on, and S3 adds the transit tier. The stationary fogs remain drop-off and recharge points in all three, so each step isolates one tier rather than changing the infrastructure count.
The stationary tier is what makes analysis tasks viable. Going from S1 to S2 raises PAD utility per drone from $3.98$ to $6.02$, a gain of $51\%$, and on-time completion from $68.7\%$ to $93.3\%$. PD is unchanged because this region has a single stationary fog, and a PD task carries no analysis to offload; the asymmetry is exactly the cost of performing inference onboard, at the expense of flight time and energy.

The transit tier then adds a further $16.1\%$ for PD and $14.4\%$ for PAD, taking completion to $100\%$ in both. Its contribution is, however, not only computational but also logistical. Fig.~\ref{fig:da-ablation-analysis} shows that buses perform $9\%$ of the PAD analysis in S3, while the onboard share is essentially unchanged from S2; what changes is that $45\%$ of PAD drop-offs and $71\%$ of PD drop-offs move to a bus. Buses, therefore, act as drop-off points for the UAV in addition to serving as a compute tier, which is why the gain appears in PD as strongly as in PAD. Consistent with this, energy per completed task falls from S2 to S3 by $17\%$, and recharge cycles per $100$ completed tasks fall by up to $31\%$ (Fig.~\ref{fig:da-ablation-swaps}); shorter flights to a nearer drop-off point leave the UAV with more charge in hand.

\subsection{Adapting to Traffic Variability}\label{sec:traffic-variability}

We now discuss the performance of ADA relative to DA under traffic variability. 
We inject a delay of $U(0, 300]$~s into the arrival of every halt, drawn once per halt and held fixed across the drones that target it. A late bus still serves its halt, so $t^h_e$ and $D(c)$ shift with $t^h_s$.
Crucially, DA still \textit{plans} on the published timetable. The delay surfaces only at execution, $\tau$ seconds before the UAV would reach the stop.
Fig.~\ref{fig:traffic-ada} compares the results for $400T50d$ over $10$ seeds. \textit{DA normal} is the undelayed baseline and \textit{DA failure} is DA meeting the delay with no recourse. 
The unprotected loss is severe. DA failure gives up $41\%$ of utility for PD and $22\%$ for PAD (Fig.~\ref{fig:utility-ada}), with completion falling from $100\%$ to $67\%$ and $82\%$ respectively. PD suffers more because it relies more heavily on the transit tier, with $62\%$ of its deliveries against $36\%$ for PAD (Fig.~\ref{fig:task-ada}). 
ADA Local recovers part of this and saturates almost at once. 
More notice does not help because the assignment was made earlier: a UAV that has flown to a distant cluster on the promise of a nearby bus cannot re-choose its way out, whatever it learns en route. 
Consistent with this, the mobile fog stack barely recovers under ADA Local, $33\%\rightarrow34\%$ for PD, and the gain instead comes from diverting to a stationary fog. ADA Global, in contrast, maintains completion above $97\%$ and recovers up to $93\%$ of the lost utility. As we observe, the mobile share \textit{rises}, $33\%\rightarrow46\%$ for PD and $18\%\rightarrow28\%$ for PAD.

\begin{figure}[t]
    \centering
    \subfloat[Average Utility per Drone.]{
    \includegraphics[width=0.4\linewidth]{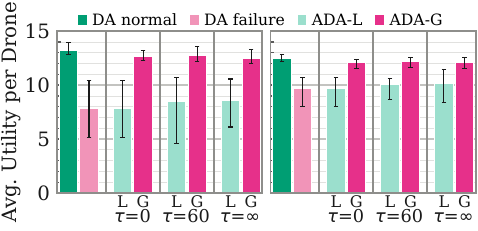}
    \label{fig:utility-ada}}
    \subfloat[Task completion; stacks show where each task was delivered.]{\includegraphics[width=0.4\linewidth]{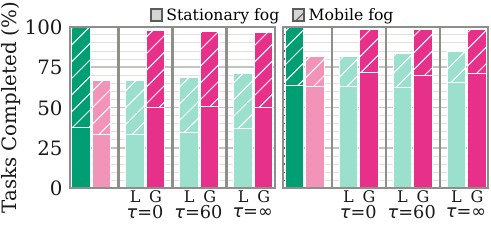}
    \label{fig:task-ada}}
    \caption{Adapting to traffic variability, 400T50d, PD (left), PAD (right).}
    \label{fig:traffic-ada}
\end{figure}

\subsection{Scheduler Runtime}\label{sec:analysis-runtime}

\begin{table}[t]
\small
\setlength{\tabcolsep}{4pt}
\caption{Median latency (s) to schedule $9$ workload each for $4$ task types, and $10$ seeds ($360$ runs per scheduler).}
\label{tab:runtime}
\centering
\begin{tabular}{l|r|r|r|r|r|r}
\toprule
& \bf{DA} & \bf{EOFO} & \bf{EDF} & \bf{NWF} & \bf{PBTO} & \bf{TGTD} \\
\midrule
Runtime (s) & $0.039$ & $0.005$ & $0.065$ & $0.068$ & $0.072$ & $0.136$ \\
\bottomrule
\end{tabular}
\end{table}

The bound in \S\ref{sec:heuristics} is dominated by the $\mathcal{O}(|\tau||\mathds{U}||\mathds{M}|)$ term.
Table~\ref{tab:runtime} reports what that costs in practice on the hardware of \S\ref{sec:evaluation}.
EOFO is the fastest scheduler at $0.005$~s, which is unsurprising: it is depot-only, so it never evaluates a halt and never pays the $|\mathds{M}|$ term. It buys that with $18$--$24\%$ less utility than DA on aggregate
(\S\ref{sec:analysis-1}). Among the schedulers that do reason about the transit tier, DA is the fastest at $0.039$~s, ahead of EDF, NWF, PBTO, and TGTD, despite being the only one that evaluates every UAV for every task; the clustering step keeps the candidate set small enough that the worst case above is never realized.

\section{Discussion}\label{sec:discussion}

Repeating our study on the $10\times10km^2$ (Fig.~\ref{fig:map-10x10}), four times the area at $8.3\%$ coverage, DA completes $57\%$ of PAD tasks, and every scheduler lands within a few points of the same ceiling. That ceiling is set by endurance rather than by scheduling: a full pack buys $1{,}800$~s of flight time at $v^u=4$~m/s, or $7.2$~km, and $43\%$ of tasks in the larger box lie far enough out that the depot-to-waypoint-to-drop-off trip cannot be flown on a single pack. Once range, rather than placement, determines the outcome, the schedulers converge. 
Restoring the platform restores the result: at $v^u=12$~m/s, or with twice the battery, completion returns to $100\%$ for DA.

Moreover, the same design and algorithms we propose can generalize to applications where deadline-bound sensing and analysis operate over regions with sparse connectivity and can exploit a timetabled vehicle that already traverses the area. For instance, powerline, rail, and pipeline inspection, post-disaster damage assessment, and forest or wildlife monitoring.

\section{Conclusion}\label{sec:conclusion}
We formulated the Mission Scheduling Problem, which co-schedules UAV routes with analytics placement across onboard, stationary and transit-borne fog, and solved it with the DA heuristic. Our work can be extended to scenarios where drones not only execute tasks but also collect data that dynamically triggers new tasks. Incorporating real-time task generation and adaptive scheduling would enhance the system’s responsiveness to evolving mission requirements.

\bibliographystyle{unsrt} 
\bibliography{references}

\end{document}